\documentclass[11pt]{article}

\pdfoutput=1

\usepackage[usenames,dvipsnames]{color}
 \usepackage{tikz}
\usetikzlibrary{trees}
\usepackage[english]{babel}
\usepackage[utf8]{inputenc}
\usepackage[T1]{fontenc}
\usepackage[pdftex, pdftitle={Article}, pdfauthor={Author}]{hyperref} 
\usepackage{mleftright}\mleftright
\usepackage[margin=1in]{geometry} 
\usepackage{graphicx}

\usetikzlibrary {positioning}
\definecolor {processblue}{cmyk}{0.96,0,0,0}

\usepackage{comment}

\usepackage{fancyhdr}
\usepackage[mmddyy,hhmmss]{datetime}
\usepackage{braket}
\usepackage{amssymb,bbm}
\usepackage{amsmath}

\usepackage{latexsym}
\usepackage{amsthm}
\usepackage[capitalise]{cleveref}

\usepackage{comment}
\usepackage{mathdots}
\usepackage{url}
\usepackage{algorithm}
\usepackage{algorithmic}
\usepackage{systeme}
\usepackage{appendix}
\hypersetup{pdfpagemode=UseNone}

\newtheorem{theorem}{Theorem}[section]
\newtheorem{lemma}[theorem]{Lemma}

\newtheorem{definition}[theorem]{Definition}

\newtheorem{problem}[theorem]{Problem}

\theoremstyle{definition}
{
	\newtheorem{remark}[theorem]{Remark}
	
}

\mathchardef\mhyphen="2D

\def\({\left(}
\def\){\right)}

\usepackage{graphicx}
\definecolor{greenn}{rgb}{0,0.8,0.2}
\definecolor{bluue}{rgb}{0.3,0,0.7}

\usepackage{ifdraft}
\ifdraft{\newcommand{\authnote}[3]{{\color{#3}{\text{ \bf #1:}} #2}}}{\newcommand{\authnote}[3]{}}

\mathchardef\mhyphen="2D
\pdfoutput=1
 
\usepackage{graphicx} 
\title{Exponential speedup of quantum algorithms for the pathfinding problem}
 
\author{ Jianqiang Li\thanks{Department of Computer Science and Engineering, Pennsylvania State University, PA, USA {\tt jxl1842@psu.edu}}}
\date{ }

\begin{document}

\maketitle

\begin{abstract}
Given $x, y$ on an unweighted undirected graph $G$, the goal of the pathfinding problem is to find an $x$-$y$ path. In this work, we first construct a graph $G$ based on welded trees and define a pathfinding problem in the adjacency list oracle $O$. Then we provide an efficient quantum algorithm to find an $x$-$y$ path in the graph $G$. Finally, we prove that no classical algorithm can find an $x$-$y$ path in subexponential time with high probability. The pathfinding problem is one of the fundamental graph-related problems. Our findings suggest that quantum algorithms could potentially offer advantages in more types of graphs to solve the pathfinding problem.
\end{abstract}

\section{Introduction}

Quantum algorithms can exploit the superposition and interference properties of quantum mechanics to solve certain problems significantly faster than classical algorithms, achieving super-polynomial speedup. However, the list of such problems is limited. The most well-known problems demonstrating this advantage include Simon's problem, factoring, discrete logarithms, the Pell equation, and other related algebraic problems. Efficient quantum algorithms \cite{simon1997SimonsAlgorithm, shor1994Factoring, hallgren2007polynomial} exist for these problems, while no known classical algorithms can solve them. In addition, there are graph-related problems, such as the welded tree problem \cite{childs2003ExpSpeedupQW} and the graph property testing problem \cite{benDavid2020SymmetriesGraphPropertiesQSpeedups}, where exponential separations can be achieved relative to the adjacency list oracle.

Identifying problems that can be leveraged by quantum algorithms to achieve superpolynomial speedup over classical algorithms remains one of the major challenges in the field of quantum computation. These challenges arise from two aspects. First, the problem needs to be believed to be classically hard. Second, the problem also needs to be quantumly easy. Several previous attempts to identify such problems have failed due to these constraints. For example, hopes of achieving exponential speedup with quantum algorithms in various machine learning problems, including recommendation systems, principal component analysis, supervised clustering, support vector machines, low-rank regression, and solving semidefinite programs, were later refuted by Tang's breakthrough results \cite{tang2018QuantumInspiredRecommSys} and subsequent work \cite{tang2018QInspiredClassAlgPCA,chia2019SampdSubLinLowRankFramework}. Moreover, although quantum algorithms \cite{chen2022quantum, ding2023limitations} that utilize the HHL algorithm \cite{HHL09} as a subroutine show potential to achieve superpolynomial speed-up in solving certain multivariate polynomial systems, it remains unknown how to find such special polynomial systems to demonstrate this quantum speedup. 


The welded tree pathfinding problem in the adjacency list oracle is one of the top open problems in the field of quantum query complexity \cite{aaronson2021open}. A welded tree graph consists of two balanced binary trees of height $n$ with roots $s$ and $t$ and a random cycle that alternates between the leaves of the two binary trees. Given the name of the two roots $s$ and $t$, the goal of the welded tree pathfinding problem is to output the names of vertices of an $s$-$t$ path.  While there is an efficient quantum algorithm to solve the welded tree problem, that is, finding the name of the root $t$ given the name of the root $s$, it has been shown that a natural class of quantum algorithms cannot solve the welded tree pathfinding problem \cite{childs2022quantum}. 


Moreover, the pathfinding problem in isogeny graphs plays an important role in the security of isogeny-based cryptosystems \cite{eisentrager2020computing, Charles2009, Costache2018,wesolowski2022supersingular}. An isogeny graph is constructed with vertices as isomorphism classes of elliptic curves and edges as isogenies (maps) between two elliptic curves. The size of the graph is exponentially large, making it difficult to find a path (map) between two vertices (elliptic curves) of polynomial length.
Depending on the constraints on the isogenies, there will be different types of graph, such as volcano graphs, Cayley graphs, and supersingular graphs. So far, the best quantum algorithm takes exponential time to solve the pathfinding problem in supersingular isogeny graphs \cite{jaques2019quantum,tani2009claw}. Notably, given a vertex in an isogeny graph, one can efficiently compute its neighbors. This makes isogeny graphs a natural instantiation of the adjacency list oracle of an abstract graph.  Recent advances have shown that efficient classical algorithms can be used to attack SIDH and SIKE in certain supersingular isogeny graphs, provided additional information is available alongside the starting vertex \cite{castryck2023efficient}. However, the general pathfinding problem underlying isogeny-based cryptosystems remains secure \cite{arpin2024orientations}. Consequently, other isogeny-based cryptosystems, such as CSIDH \cite{castryck2018csidh} continue to be secure against this attack \cite{Galbraith2022}.

In this paper, we show that the pathfinding problem in some graphs admits an exponential separation between the quantum algorithm and the classical algorithm under the adjacency list oracle. Instead of solving the welded tree pathfinding problem and the isogeny graph pathfinding problem directly, we construct a graph $G$ by associating $n$ distinct welded trees with a path of length $n$ as in Figure \ref{fig:expweldedtree}. Given the name of two vertices $x$,$y$ in $G$, the goal of the pathfinding problem is to output the names of vertices of an $x$-$y$ path. Similarly to the welded tree pathfinding problem, there is an exponential number of $x$-$y$ paths in the graph $G$. On the other hand, the shortest $x$-$y$ path is unique in the graph $G$, while there is an exponential number of shortest paths in the welded tree graph. Using the distinctness of the welded trees in the graph $G$, we show that there is an efficient quantum algorithm that solves this pathfinding problem. This quantum algorithm works by finding the edges of the $x$-$y$ shortest path step-by-step. For each step, the polynomial time continuous quantum walk algorithm \cite{childs2003ExpSpeedupQW} for the welded tree problem is used to select one edge of the shortest $x$-$y$ path. Repeat $n$ steps, the quantum algorithm outputs the shortest $x$-$y$ path. In particular, our quantum algorithm uses the continuous quantum walk algorithm \cite{childs2003ExpSpeedupQW} as a subroutine, which is a new class of algorithm outside of the natural quantum algorithms considered in \cite{childs2022quantum}. This suggests that there might be a new way to solve the welded tree pathfinding problem.

Finally, we show that no classical algorithm can solve the path-finding problem in subexponential time. Observe that any paths between the two vertices $s$ and $t$ have to go through a vertex of degree at least $4$, or pass the vertex $p_{n/2}$ as indicated in Figure \ref{fig:expweldedtree}. To establish the classical lower bound, we show that it is classically hard to output the name of the vertex $p_{n/2}$ or vertices that have a degree of at least $4$. The proof uses the result of the classical lower bound of the welded tree problem, which informally states that, given the name of one root, it is classically hard to traverse through the welded tree to find the other root.

To the best of our knowledge, this is the first example that exhibits the exponential speedup of quantum algorithms for the pathfinding problem. Previous quantum algorithms have achieved at most a polynomial speedup for solving the pathfinding problem in both general and specific graphs. For example, quantum algorithms based on amplitude amplification \cite{durr2006quantum} can solve the pathfinding problem in a general graph using $O(N^{3/2})$ queries to the adjacency matrix oracle, where $N$ is the number of vertices in the graph. By restricting the class of graphs and under the same query model, a quantum algorithm \cite{jeffery2023quantum} based on the span program to generate the quantum $s$-$t$ electrical flow slightly improves the query complexity to solve the pathfinding problem.  In the incidence list oracle model, \cite{Seaneletr} uses the HHL algorithm to generate the quantum $s$-$t$ electrical flow and provide a sublinear time ($\Tilde{O}(\sqrt{N})$) quantum algorithm for the pathfinding problem in some graphs while the known classical and quantum algorithms take linear time $\Omega(N)$.
Furthermore, there exist quantum algorithms \cite{reitzner2017finding,hillery2021finding,koch2018finding} that utilize quantum walks to solve the pathfinding problem in specific graph structures such as a chain of star graphs and a regular tree graph, providing a quadratic speedup compared to classical algorithms.


\paragraph{Organization.} In Section \ref{sec:qaweldedtree}, we introduce the main results of the continuous quantum walk for the welded tree problem, which will be a key subroutine of our quantum pathfinding algorithm. In Section \ref{sec:gralg}, we construct the graph $G$ from the welded trees and use it to define the pathfinding problem. Then we present the quantum algorithm that solves the pathfinding problem. In Section \ref{sec:claslower}, we establish the classical lower bound of the pathfinding problem showing that no classical algorithm can solve this problem with a polynomial number of queries. Finally, we give a discussion and provide some open problems in Section \ref{sec:concu}.

\section{Quantum algorithms for the welded tree problem} \label{sec:qaweldedtree}
In this section, we introduce the main results of the continuous quantum walk approach for the welded tree problem, which will be a key subroutine of our quantum algorithm for the pathfinding problem.

It should be noted that several quantum algorithms have been developed to solve the welded tree problem, encompassing techniques such as continuous quantum walks \cite{childs2003ExpSpeedupQW}, multidimensional quantum walks \cite{jeffery2023multidimensional}, discrete quantum walks \cite{li2023recover} and the coupled classical oscillator approach \cite{babbush2023exponential}. However, in this paper, we specifically focus on the original continuous quantum walk approach for the welded tree problem.

A welded tree graph $W$ consists of two balanced binary trees of height $n$ with roots $s$ and $t$ and a random cycle that alternates between the leaves of the two binary trees. The number of vertices in $W$ is $2^{n+2}-2$ and the names of the vertices are randomly assigned from the set $\{0,1\}^{2n}$. To access the neighbors of a particular vertex, we use an adjacency list oracle denoted as $O$ for the graph $G$.  Given a $2n$-bit string $u \in \{0,1\}^{2n}$, the adjacency list oracle $O$ provides the neighboring vertices of $u$, or it returns $\perp$ if $u$ is not a valid vertex name in the graph.

\begin{problem}[The welded tree problem] Given an adjacency list oracle $O$ of the welded tree $W$ and the name of the starting vertex $s \in \{0,1\}^{2n}$, the goal is to output the name of the other root $t$.
\end{problem}

Let $A$ be the adjacency matrix of a graph $G$, $\tau$ be a real number, and $\ket{\psi_0} =\ket{s}$ as the initial input state, the continuous-time quantum walk is defined as

\[ 
\ket{\psi_\tau} = e^{-iA\tau}\ket{\psi_0}.\]


\begin{lemma}[Theorem 3 in \cite{childs2003ExpSpeedupQW}] \label{lem:weldedcqw}
Let $A$ be the adjacency matrix of the welded tree graph $W$. With the adjacency list oracle $O$ of the welded tree $W$, the name of the starting vertex $s\in \{0,1\}^{2n}$ and let $\ket{\psi_0}=\ket{s}$, running the continuous quantum walk $\ket{\psi_\tau} = e^{-iA\tau}\ket{\psi_0}$ for a time $\tau$ chosen uniformly in $[0, n^5]$ and then measuring on the computational basis yields a probability of finding the name of the other root vertex $t$ that is greater than $\Omega(\frac{1}{n})$.
\end{lemma}

By querying the adjacency list oracle $O$ of a given vertex, one can determine whether it has a degree of $2$ or not. Since only two vertices, $s$ and $t$, have a degree of $2$ in the welded tree $W$, repeating the quantum algorithm multiple times can find the other root $t$ in poly($n$) time. On the other hand, no classical algorithm can solve the welded tree problem in subexponential time by the following \cref{lem:weldedlowerbound}. 
\begin{lemma}[Theorem 9 in \cite{childs2003ExpSpeedupQW}] \label{lem:weldedlowerbound} 
 For the welded tree problem,  
any classical algorithm that makes at most $2^{n/6}$ queries to the oracle finds the ending vertex or a cycle with probability at most $4\cdot 2^{-n/6}$.
\end{lemma} 

Note that the exact statement of Theorem 9 in \cite{childs2003ExpSpeedupQW} does not mention finding a cycle while the proof in their classical lower bound result holds for this case.

\section{The pathfinding problem and the quantum algorithm} \label{sec:gralg}
In this section, we construct a graph $G$ based on welded trees and provide an efficient quantum algorithm that outputs an $x$-$y$ path. 
The key idea is to construct the graph $G$ associating an $x$-$y$ path $P_n$ of length $n$ with $n$ distinct welded trees. The quantum algorithm finds the $x$-$y$ path step by step. For each step, we use the quantum algorithm for the welded tree problem \cite{childs2003ExpSpeedupQW} to detect and select the edge $(x,u_i)$ that is in the path $P_n$. Then we remove the selected edge and update the starting vertex $x = u_i$. Repeat until $x=y$ and output all selected edges as an $x$-$y$ path.


Without loss of generality, assume that $n$ is an even positive integer. Given $n$ disjoint welded trees $W^i$ with roots $s_i$ and $t_i$, the graph $G$ is constructed as follows:

\begin{enumerate}
    \item For each $i\in [n]$, adding $i+1$ new isolated edges $e_k=(m^i_k,n^i_k), k \in [i+1]$,   then adding $i+1$ new edges between $t_i$ and vertices $m^i_k$  such that the root vertex $t_i$ has degree $i+3$.

    
    \item Given a path graph $P_n$ with $n$ vertices $p_1, \cdots, p_n$, for each $i\in [n]$, add an edge between $p_i$ and $s_{i}$ . 

    \item Adding a random cycle between all vertices $n^i_k$ for each $i\in [n]$ and $k\in [i+1]$.
     Denote the resulting graph as $G$. 
\end{enumerate}

\begin{figure}[H]
\centering
\includegraphics[width=15cm,height=8cm]{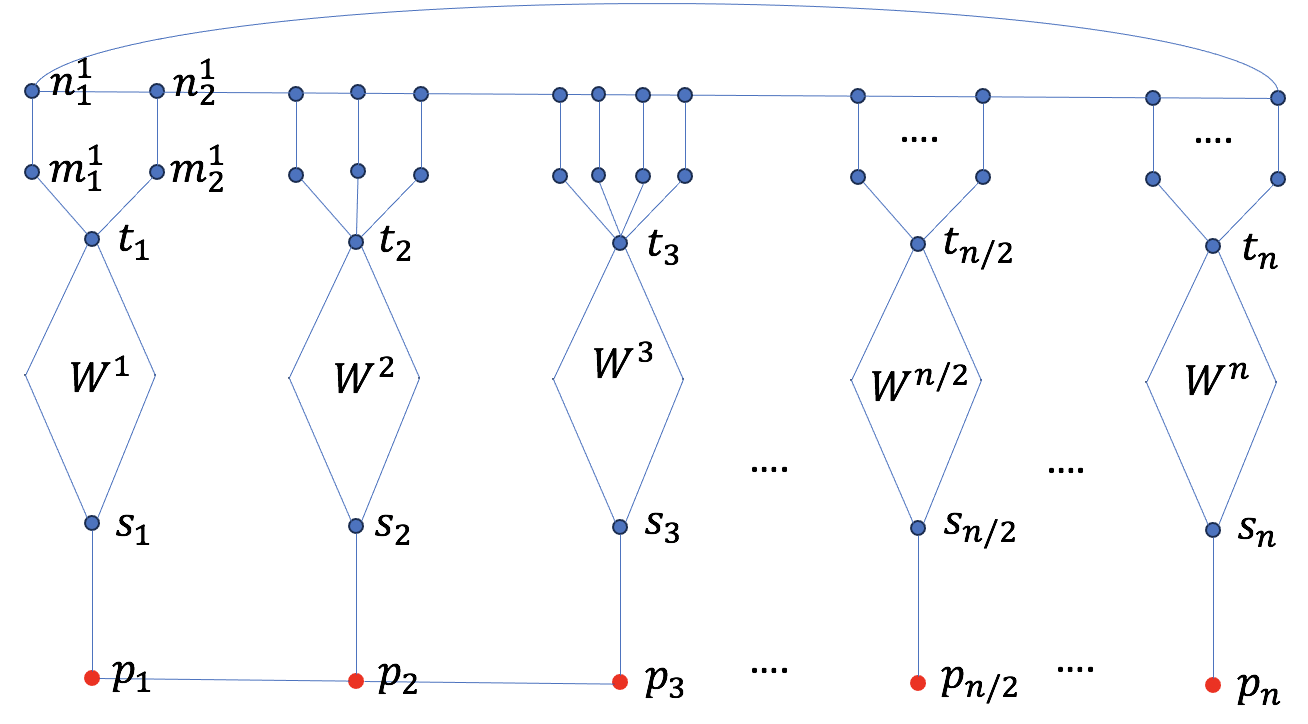}
\caption{The Welded Tree Path Graph $G$}
\label{fig:expweldedtree}
\end{figure}


Let $x=p_1$ be the starting vertex, $y=p_n$ be the ending vertex, and given the names of $s$ and $t$, the goal is to find an $x$-$y$ path in the graph $G$.
 \begin{problem} [Pathfinding problem] \label{prob:pathfind}Given the adjacency list oracle $O$ of the graph $G$ and the name of the starting, ending vertices $x, y\in \{0,1\}^{2n}$, the goal is to output an $x$-$y$ path in $G$.\end{problem}

\begin{algorithm}
\caption{Quantum algorithm for finding an $s$-$t$ path in the graph $G$}
\begin{algorithmic}\label{alg:weldedqwpath}
\REQUIRE Graph $G=(V,E)$, $x, y\in V$ and $i=1$.
\ENSURE  an $x$-$y$ path
\begin{enumerate}
 
  \STATE Given the name of the vertex $x \in \{0,1\}^{2n}$, the adjacency list oracle $O$  returns the names of the two neighbors of $x$, that is $u_1, u_2 \in \{0,1\}^{2n}$.  Without loss of generality, pick one of the two neighbors $u_1$ as the initial state $\ket{u_1}$. 
  
 \STATE \label{ite:algroithm} Let $A'$ be the modified adjacency matrix of the adjacency matrix $A$ of $G$ by removing all the edges adjacent to a degree $2$ vertex. Let $O'$ be $O$ except it returns no edge if one endpoint of the edge is a degree $2$ vertex. Run the continuous quantum walk  $e^{-iA'\tau}\ket{u_1}$ for a uniform random time $\tau \in [0, n^5]$.
 Measure the resulting state in the computational basis and get the name of an outcome vertex $v \in \{0,1\}^{2n}$. Compute the degree of the vertex $v$ by querying the adjacency list oracle $O$ of $G$. 
 
 \STATE Repeat Step ~\ref{ite:algroithm} $n^2$ times or until the degree of the measured vertex equals $i+3$. 
  
  \STATE   If the degree of the measured vertex $v$ equals $i+3$, then collect the edge $(x,u_2)$ as an $x$-$y$ path edge and let $x=u_2$.  Otherwise, collect the edge $(x,u_1)$ as an $x$-$y$ path edge and let $x=u_1$. 
  
 \STATE Let i=i+1 and update the graph $G$ by deleting the selected edge.  Repeat all the above steps until $x=y$, then output all the selected edges as an $x$-$y$ path.
   \end{enumerate}
\end{algorithmic}
  \end{algorithm} 
\begin{theorem} With high probability, Algorithm \ref{alg:weldedqwpath} outputs an $x$-$y$ path in poly($n$) time.  
\end{theorem}
\begin{proof}
 For the $i$-th iteration, the key subroutine is to implement the continuous quantum walk $e^{-iA'\tau}$ for $\tau \in [0,n^5]$ in Step 2.  The adjacency list oracle $O'$ can be performed by removing all edges adjacent to a degree $2$ vertex from the adjacency list oracle $O$. Then use the adjacency list oracle $O'$ to construct a block encoding of the modified adjacency matrix $A'$. With this block encoding, we can perform the Hamiltonian simulation  $e^{-iA'\tau} \ket{u_1}$ in poly($n$) time by implementing polynomials of $A'$ \cite[Theorem 58]{gilyen2018quantum}. 

Recall that the only vertex that has degree $i+3$ is the root $t_i$ of the welded tree $W^i$ in the graph $G$. If $u_1$ is the same as the root vertex $s_i$ of the welded tree $W^i$, then by Lemma \ref{lem:weldedcqw}, the success probability of measuring a vertex of degree $i+3$ is $\Omega(1/n)$. The probability of success is $1 - \exp(-\Omega(n))$ by repeating Step 2 $n^2$ times. Otherwise, the probability of measuring a vertex of degree $i+3$ is $0$. This is true because the initial starting vertex $\ket{u_1}$ is disconnected from the welded tree $W^1, \cdots, W^i$ and connected to the welded trees $W^{i+1},\cdots W^{n}$ in the modified graph associated with the adjacency list oracle $O'$. Therefore, for each iteration, Algorithm \ref{alg:weldedqwpath} outputs an edge of the $x$-$y$ path with probability $1 - \exp(-\Omega(n))$ in poly($n$) time.

The length of the $x$-$y$ path outputted by Algorithm \ref{alg:weldedqwpath} is $n$, the same as the total number of iterations. Therefore, with high probability, Algorithm \ref{alg:weldedqwpath} outputs an $x$-$y$ path in poly($n$) time. 
\end{proof}

\begin{remark}
  The random cycle at the top of the graph $G$ is not strictly necessary for showing the exponential quantum-classical separation for the pathfinding problem. Instead, it serves to illustrate the exponential separation persists when there are multiple $x$-$y$ paths, highlighting the robustness of the quantum advantage in more complex graph structures.
\end{remark}

\section{Classical Lower bound}\label{sec:claslower}

In this section, we prove that no classical randomized algorithm $\mathcal{R}$ can solve the pathfinding problem \cref{prob:pathfind} in the welded tree path graph $G=(V,E)$ within subexponential time. Our approach is inspired by the classical lower bound established for the welded (glued) tree graph in \cite{childs2003ExpSpeedupQW}, which demonstrates an exponential quantum-classical separation for finding a marked vertex. We adapt their techniques to address the classical lower bound for the pathfinding problem.

Intuitively, the pathfinding problem can be seen as a generalization of the vertex-finding problem, as it requires identifying a sequence of vertices that form an $x$-$y$ path. To establish a classical lower bound for the pathfinding problem, we demonstrate that no classical algorithm can efficiently identify certain critical vertices along 
$x$-$y$ paths of the constructed welded tree path graph $G=(V,E)$. By the construction of the graph $G$, we define the set of critical vertices as $p_{n/2}$ and the vertices with degrees of at least $4$, specifically 
$\{t_1,\ldots,t_n\}$.

First, by assigning a random $2n$-bits string to each vertex in the graph $G$, any classical algorithm to solve the pathfinding problem with at most $2^{n/6}$ queries of $O$ must traverse connected subgraphs with starting vertices $x$ or $y$. The reason is as follows: the total number of vertices in the graph $G$ is $n2^{n+2}+\frac{n(n+1)}{2}$ while the total number of potential vertex names is $2^{2n}$. If $\mathcal{R}$ makes at most $2^{n/6}$ queries, the probability of querying the name of a vertex that has not been previously returned by the oracle $O$ is at most $$2^{n/6} (n2^{n+2}+\frac{n(n+1)}{2})/2^{2n} = O(2^{-n/6}).$$  Thus, it is unlikely that $\mathcal{R}$ can guess a name of vertices of the graph $G$ that is not returned by the oracle $O$ when making at most $2^{n/6}$ queries. Without loss of generality, in the rest of the section, we restrict $\mathcal{R}$ to traverse two connected subgraphs starting with vertices $x$ and $y$, respectively. 

Under these conditions, solving the pathfinding problem with at most $2^{n/6}$ queries of the adjacency list oracle $O$ is equivalent to playing the following game.

\vspace{1mm}
\textbf{Game A} 
The oracle $O$ assigns each vertex of a randomly chosen graph $G$ a distinct $2n$-bit string as its name. The starting vertex $x$ has the name $0^{2n}$ and the ending vertex $y$ has the name $1^{2n}$. At each step, $\mathcal{R}$ sends a $2n$-bit string to $O$, and the oracle $O$ returns the names of the neighbors of that vertex if the given vertex name is valid. $\mathcal{R}$ wins if it outputs the names of the vertices forming an $x$-$y$ path.

To bound on the success probability of $\mathcal{R}$ winning Game A, we consider a simpler Game B, which relaxes the winning conditions. In Game B, $\mathcal{R}$ wins by outputting the name of any critical vertex, such as  $\{t_1,t_2,\ldots,t_n\}$ or the vertex $p_{n/2}$ or a cycle while winning Game A needs to output names of an $x$-$y$ paths.
\vspace{1mm}

\textbf{Game B} Let Game B be the same as Game A, except that $\mathcal{R}$ wins if it outputs the name of any vertex of degree at least $4$, i.e., $\{t_1,t_2,\ldots,t_n\}$ or the name of the vertex $p_{n/2}$, or a cycle among the vertices visited by $\mathcal{R}$.

 To prove the classical lower bound for the pathfinding problem, it suffices to show that $\mathcal{R}$ cannot win this easier Game B in subexponential time. Following \cite{childs2003ExpSpeedupQW}, the condition that allows $\mathcal{R}$ to win  by finding a cycle in Game B simplifies the analysis. One key concept in \cite{childs2003ExpSpeedupQW} is the random embedding of a rooted binary tree. To analyze the success probability of $\mathcal{R}$ in Game B, we restate the defintion of a random embedding of a rooted binary tree as follows:
 
 


\begin{definition}[Random embedding of a rooted binary tree] Given the name of the starting vertex $x$, the random embedding of a rooted binary tree $T$ into the graph $G$ is defined as a function $\pi$ from the vertices of $T$ to the vertices of $G$ such that \begin{enumerate}
    \item $\pi(\textit{ROOT}) = x$.
    \item Let $i$ and $j$ be the two neighbors of $ROOT$ in $T$ and let $u$ and $v$ be the neighbors of $x$ in $G$. With probability $1/2$ set $\pi(i)= u$ and $\pi(j)=v$, and with probability $1/2$ set $\pi(i)= v$ and $\pi(j)=u$. 
\item For any non-leaf vertex $i$ in $T$, let $j$ and $k$ denote the children of vertex $i$, and let $\ell$ denote the parent of vertex $i$.    Let $u$ and $v$ be the two neighbors of $\pi (i)$ in $G$ other than $\pi (\ell)$. With probability $1/2$ set $\pi(i)= u$ and $\pi(j)=v$, and with probability $1/2$ set $\pi(i)= v$ and $\pi(j)=u$.
\end{enumerate}
  We say that an embedding $\pi$ is proper if it is injective, that is for any $i,j\in T$, we have $\pi(i)\neq \pi(j)$ if $i\neq j$.  We say that $T$ exits under $\pi$ if, for any $i\in T$, the degree of $\pi(i)$ in $G$ is at least $4$ or $\pi(i)$ is $p_{n/2}$ or $y=p_{n}$.   
\end{definition}



\begin{theorem}
\label{lem:clascilowGameC} 
 If $\mathcal{R}$ uses $2^{n/6}$ queries to the oracle $O$, then its probability of finding an $s$-$t$ path is at most $2(n+1)\cdot 4\cdot 2^{-n/6}$. 
\end{theorem}
\begin{proof}  To obtain an upper bound of the probability of $\mathcal{R}$ for solving the pathfinding problem, it suffices to show that the probability of $\mathcal{R}$ finding the certain set of vertices, that is  $p_{n/2}$ and vertices that have degrees at least $4$,  is at most $2(n+1)\cdot 4\cdot 2^{-n/6}$.
\begin{enumerate}
    \item

Given the name of the starting vertex $s$, the upper bound of the probability of $\mathcal{R}$ winning Game $B$ can be expressed as the probability that the embedding $\pi$ is improper or $T$ exits $G$ under $\pi$ using the result of \cite[Equation 72, Lemma 7]{childs2003ExpSpeedupQW}. This comes from the fact that there exists an algorithm $\mathcal{R}'$ to generate a random rooted binary tree by simulating any classical algorithm $\mathcal{R}$ for Game B.

Let $T$ be a random rooted binary tree with at most $2^{n/6}$ vertices and $\pi(T) \subseteq V$ be the image in the graph $G$ under the random embedding $\pi$. We have the following result:

\begin{enumerate}
    \item First, it is unlikely that $\pi(T)$ contains a vertex with degree at least $4$, i.e., $\{t_1,t_2,\ldots,t_n\}$ or a cycle.  For the vertices $\{t_1,t_2,\ldots,t_n\}$ and the cycles in $G$ can only be located by entering a welded tree graph. This is true because, as indicated in Figure \ref{fig:expweldedtree}, the vertices that have a degree at least $4$ are the roots $t_i$, and the path $P_n$ with the $n$ welded trees form a tree structure. Lemma~\ref{lem:weldedlowerbound} shows that any classical algorithm that makes at most $2^{n/6}$ queries to the oracle $O$ can find a root $t_i$ or a cycle in a single welded tree graph $W^i$ with probability at most $ 4\cdot 2^{-n/6}$. Since there are $n$ welded trees and $n$ vertices of degree at least $4$ in the welded tree path graph $G$, the union bound \footnote{The union bound states that the probability of at least one event occurring is at most the sum of the probability of each event: $\Pr[\text{at least one success}] \leq \sum_{i=1}^{n} \Pr[\text{success in tree $i$}]$.} implies that the 
    \[
    \Pr[\mathcal{R} \text{ finds a cycle or $t_i, 1 \leq i\leq n$}] \leq n\cdot 4 \cdot 2^{-n/6}.
    \]

 \item The probability of $\pi(T)$ contains the vertex $p_{n/2}$ is exponentially small.
 Consider a path in $T$ from the root to a leaf. To reach the vertex $p_{n/2}$, $\pi$ must follow the path $P_n$ $\frac{n}{2}$ times, which has probability $2^{-n/2}$. Since there are at most $2^{n/6}$ tries on each path of $T$ and there are at most $2^{n/6}$ paths. The probability of finding the name of the vertex $p_{n/2}$ is at most $2^{-n/6}$, that is,
 \[
 \Pr[\mathcal{R} \text{ finds } p_{n/2}] \leq 2^{n/6} \cdot 2^{n/6} \cdot 2^{-n/2}=2^{-n/6}.
 \]

 \end{enumerate}
 Therefore, by the union bound, given the name of the starting vertex $s$, the probability of $\mathcal{R}$ finding the vertex $p_{n/2}$, a cycle, or a vertex of degree at least $4$ is at most $ (n+1) \cdot 4\cdot 2^{-n/6}$.
\item By symmetry structure of the welded tree path graph,, the same bound also holds when the only given name is the vertex $y$. 
\end{enumerate}
Hence, given the names of $x$ and $y$ and by union bound, the probability of $\mathcal{R}$ uses  $2^{n/6}$ queries to the oracle $O$ to find a cycle, the vertex $p_{n/2}$, or a vertex with a degree at least $4$ is at most $2(n+1) \cdot 4\cdot 2^{-n/6}$. That is, $\mathcal{R}$ cannot win Game A in subexponential time.
\end{proof}

\section{Discussion} \label{sec:concu}
In this paper, we show that the pathfinding problem in the graph $G$ admits an exponential separation between quantum and classical algorithms under the adjacency list oracle. The key idea is to encode the path information into $n$ distinct welded trees, in which quantum algorithms can extract the path information efficiently by distinguishing the $n$ distinct welded trees but no classical algorithm can do this in polynomial time. In fact, there are more types of graphs to show an exponential separation of the pathfinding problem as long as they satisfy this property, for example, replacing the random cycle at the top of the graph $G$ with an arbitrary graph. Also,  using the continuous quantum walk approach, \cite{balasubramanian2023exponential} exhibits a superpolynomial quantum-classical separation to find a vertex in random hierarchy graphs. Therefore, we can also replace the welded trees with random hierarchy graphs to achieve quantum-classical separation for the pathfinding problem.

Moving forward, we provide several open problems related to the quantum advantages of pathfinding problems in more types of graphs that do not have the properties used in this paper.

\begin{itemize}
    
    \item The welded tree path graph $G$ we constructed in our paper is not regular, which is the structure that our quantum algorithm exploits to run efficiently. However, isogeny graphs, such as the Cayley graph and supersingular graphs, are examples of regular graphs. It would be intriguing to explore whether there exist regular graphs that exhibit exponential separations between quantum and classical algorithms.

   \item In addition to achieving exponential speedups, exploring and characterizing polynomial speedups of quantum algorithms for the pathfinding problem in various types of graphs is a promising and exciting future direction.
   

\end{itemize}

\section*{Acknowledgements}
J.L.\ would like to thank Sean Hallgren, Andrew Childs, Yi-Kai Liu, Daochen Wang, Sebastian Zur and Guanzhong Li for their valuable feedback and conversations and Mingming Chen for assistance in drawing the figure. Part of this work was done while the author visited Quantum Information and Computer Science (QuICS) and the Simons Institute for the Theory of Computing.  J.L. acknowledges funding from the National Science Foundation awards CCF-1618287, CNS-1617802, and CCF-1617710, and a Vannevar Bush Faculty Fellowship from the US Department of Defense.

\bibliographystyle{unsrt}

\bibliography{pathfinding2,Texport,extra}

\end{document}